\documentclass[conference]{IEEEtran}
\IEEEoverridecommandlockouts

\usepackage{geometry}
\usepackage{balance}
\usepackage{multirow}
\usepackage{lipsum}
\usepackage{amsfonts}
\usepackage{array}
\usepackage{amsmath,cases,amssymb}
\usepackage{amsmath}
\usepackage{amsthm}
\usepackage{graphicx}
\usepackage{cite}
\usepackage{subfigure}
\usepackage{epsfig}
\usepackage{url}
\usepackage{algorithm}
\usepackage{algorithmic}
\usepackage{epstopdf}
\usepackage{color}
\usepackage{makecell}
\usepackage{slashbox}
\DeclareGraphicsExtensions{.eps,.pdf}

\newcommand{\bU}{\mathbf{U}}
\newcommand{\bo}{\mathbf{o}}
\newcommand{\bmu}{\pmb{\mu}}

\newcommand{\Bcal}{\mathcal{B}}

\newcommand{\Kcal}{\mathcal{K}}

\newcommand{\UE}{\mathtt{UE}}
\newcommand{\UEk}{\mathtt{UE}_k}

\newcommand{\Hcal}{\mathcal{H}}
\newcommand{\Xcal}{\mathcal{X}}

\makeatletter
\newcommand{\doublewidetilde}[1]{{%
		\mathpalette\double@widetilde{#1}}}
\newcommand{\double@widetilde}[2]{%
	\sbox\z@{$\m@th#1\widetilde{#2}$}%
	\ht\z@=.5\ht\z@
	\widetilde{\box\z@}}
\makeatother

\newcommand*{\hili}{\color{black}}

\newtheorem{theorem}{Theorem}
\newtheorem{lemma}{Lemma}

\newtheorem{example}{Example}

\begin{document}
	
	\title{\huge Joint Beam Placement and Load Balancing Optimization for Non-Geostationary Satellite Systems\vspace*{-0.4cm}}
	
	\author{Van-Phuc~Bui$^*$, Trinh~Van~Chien$^\xi$,  Eva~Lagunas$^*$,  Jo\"el~Grotz$^\dag$, Symeon~Chatzinotas$^*$,  and Bj\"orn~Ottersten$^*$ \vspace*{-0.1cm}
		\\
		$^*$Interdisciplinary Centre for Security, Reliability and Trust (SnT), University of Luxembourg, Luxembourg\\
		$^\xi$School of Information and Communication Technology, Hanoi University of Science and Technology,  Vietnam\\ 
		$^\dag$SES Engineering, Luxembourg, Luxembourg
		\vspace*{-0.6cm}
		\thanks{This   work   was    supported   by  the   Luxembourg   National Research  Fund  (FNR)  under  the  project INtegrated Satellite - TeRrestrial Systems for Ubiquitous Beyond 5G CommunicaTions (INSTRUCT).}
		
	}
	\maketitle
	%	\vspace*{-1.5cm}
	\begin{abstract}
		Non-geostationary (Non-GSO) satellite constellations have emerged as a promising solution to enable ubiquitous high-speed low-latency broadband services by generating multiple spot-beams  placed on the ground according to the user locations. However, there is an inherent trade-off between the number of active beams  and the complexity of generating a large number of beams. This paper formulates and solves a joint beam placement and load balancing  problem to carefully optimize the satellite beam and enhance the link budgets with a minimal number of active beams. We propose a two-stage algorithm design to overcome the combinatorial structure of the considered optimization problem providing a solution in polynomial time. The first stage minimizes the number of active beams, while the second stage performs a load balancing to distribute users in the coverage area of the active beams. Numerical results confirm the benefits of the proposed methodology  both in carrier-to-noise ratio and multiplexed users per beam over other benchmarks. 
	\end{abstract}
	%\begin{IEEEkeywords}
	%	Multi-beam high throughput satellite communications, quality of service requirements, multi-objective optimization, neural networks.
	%\end{IEEEkeywords}
	
	\vspace{-0.2cm}
	%%%%%%%%%%%%%%%%%%%%%%%%%%%%%%%%%%%%%%%%%%%%%%%%
	\section{Introduction}\label{sec:intro}
	\vspace{-0.1cm}
	%%%%%%%%%%%%%%%%%%%%%%%%%%%%%%%%%%%%%%%%%%%%%%%%
	Non-geostationary (Non-GSO) satellite communication systems are expected to play an important role in future global wireless communications \cite{9473786}. Both low Earth orbit (LEO) and medium Earth orbit (MEO) constellations are revolutionizing the satellite broadband market \cite{SES_MEO}. The main benefit of such networks resides mainly on their proximity to Earth, which translates into much lower latency than their GSO satellite counterparts. These constellations are equipped with the latest advances in payload and antenna design, including radio resource adaptability and beamforming capabilities \cite{9237970}.

	Resource allocation problems in satellite communication typically include, for example,  $i)$ beam footprint design; $ii)$ user-to-beam assignment; and $iii)$ allocating limited radio resources  to each satellite beam such that users' demands should be satisfied. In the literature, several works have investigated the resource allocation problems for the multi-beam satellite systems comprising power control \cite{van2021user}, frequency assignment \cite{mizuike1989optimization}, joint power and bandwidth allocations \cite{pachler2020allocating}, and application of deep reinforcement learning  \cite{hu2018deep}. Regrading beam footprint design, certain works have aimed at a fully-flexible beamforming design \cite{9610022} which is not practical in terms of complexity, particularly for the fast pass times of LEO and MEO satellites. Herein, we focus on a  practical scenario which is based on  the actual functioning mode of operational Non-GSO satellites.  {\hili In particular, we consider the beam placement problem, where the conical-shaped beams are predefined in half power beam width (HPBW) and its location on the coverage area needs to be defined.}   The beam placement problem optimizes the number of users to  each satellite beam together with the beam center. Recently, the different beam placement problems have been studied for LEO constellations in \cite{pachler2021static}. These related works  proposed  heuristic algorithms to find the number of active beams to strategically and effectively use satellite resources. Although the proposed algorithms can find the number of active beams at the satellite, an unbalanced load appears between the beams. Specifically, the significantly different number of users among the active beams results in a heavy load of a few beams. The unbalanced load becomes severe when many users simultaneously request to access  some particular beams, while the other active beams are being in a sparse situation with only a few users. To overcome this issue, linear  techniques for the beam layout design in multi-beam satellite systems were proposed in \cite{camino2019linearization}. The multi-spot beam arrangement approach was studied in \cite{takahashi2020adaptive} to determine how the distances between two spot beams enhances the system performance. These works have not classified users into individual beams and specified the active beam number.
	
	Inspired by the above discussions, this paper formulates a weighted sum minimization problem for the beam placement and load balancing issue in the Non-GSO satellite systems to jointly optimize the required number of active beams and the connection quality between the users and the satellite. The main contributions of this paper are summarized as: $(i)$ We formulate a novel optimization problem to  minimize the distance from the users to their served beam center with the minimal number of active satellite beams. We emphasize that this non-convex and NP-hard optimization has not been investigated in the literature yet; $(ii)$ We propose a two-stage algorithm for efficient and practical implementations in polynomial time:  We first propose as algorithm to find the minimum number of active beams  required to to cover all the users at least with the HPBW. Next, we refine the solution of the first step targeting a balanced number of user per beam; and $iii)$ We provide experimental results to demonstrate the effectiveness of the proposed algorithm and compare with available benchmarks in the literature.
	
	\vspace*{-0.1cm}
	\section{Multi-Beam Satellite System Model}\label{sec:sys_model}
	\vspace*{-0.1cm}
	\begin{figure}[t]
		\centering
		\includegraphics[width= 0.28 \textwidth]{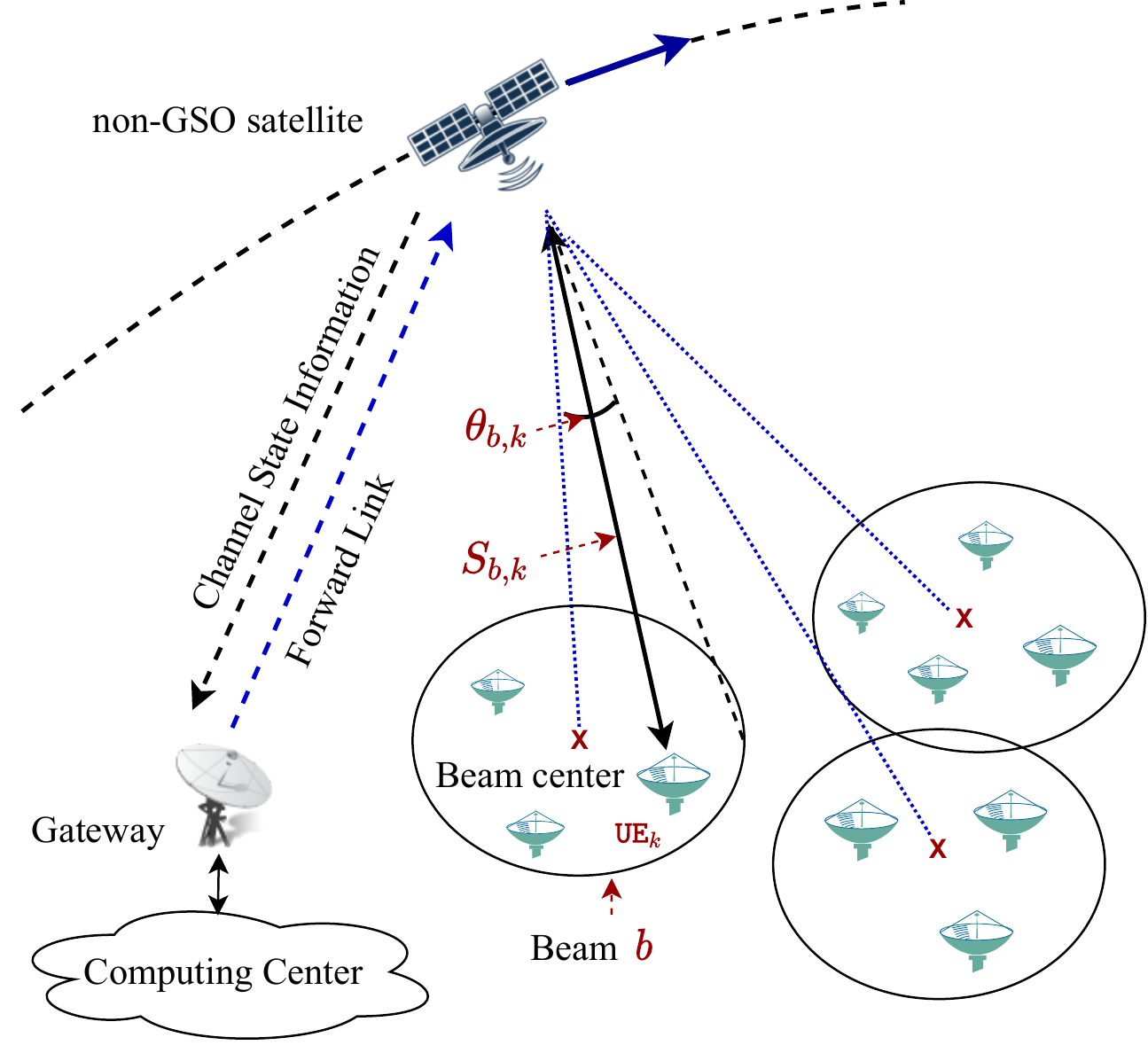}
%		\vspace*{-0.3cm}
		\caption{The multi-beam multi-user satellite system model.}
		\label{fig:system_model}
		\vspace*{-0.45cm}
	\end{figure}
	A multi-beam multi-user Non-GSO satellite system is considered in this paper, as illustrated in Fig.~\ref{fig:system_model}. The gateway located on the ground is connected to the satellite through an ideal feeder link and the beam placement optimization is handled at the gateways.  Specifically, a satellite equipped with an antenna array that enables to create a maximum of $N$ beams to serve $K$ users in the coverage area with $N \geq K$. Due to a limited power budget at the satellite and mutual interference especially at the beam boundary, it is not always beneficial to activate all the beams. Consequently, let us denote $\mathcal{B} \subseteq \{1, \ldots, N\}$ the set of active satellite beam indices with $|\mathcal{B}| \leq N$. In addition, we denote  $\mathcal{K} = \{1, \ldots,  K\}$ the set of user indices. With a proper scheduling design, the $K$ users are assigned into the active satellite beams for which $\mathcal{K}_b \subseteq \mathcal{K}, b \in \mathcal{B},$ contains the indices of the users served by the $b$-th active beam. The following property is established as
	\vspace*{-0.2cm}
	\begin{equation}
		\mathcal{K}_b \neq \emptyset, \mathcal{K}_b \cap \mathcal{K}_{b'} = \emptyset, \forall b, b' \in \mathcal{B}, \mbox{ and } \bigcup\nolimits_{b=1}^{|\mathcal{B}|} \mathcal{K}_b = \mathcal{K},
		\vspace*{-0.1cm}
	\end{equation}
	which indicates that each beam should serve at least one user. Let us denote $\theta_{b,k}$ ($0 \leq \theta_{b,k} \leq \pi/2$) the angle between the spot beam center of the $b$-th beam and the $\UEk$'s location, $k \in \mathcal \mathcal{K}_b,$ as seen from the satellite.  Then, the radiation pattern from the $b$-th beam to $\UEk$ can be mathematically formulated as a function of $\theta_{b,k}$, which is reported in the 3GPP \cite{3gpp2018study} as
	\begin{figure}[t]
	\centering
	\includegraphics[trim=0.6cm 0.0cm 1.4cm 0.6cm, clip=true, width=2.8in]{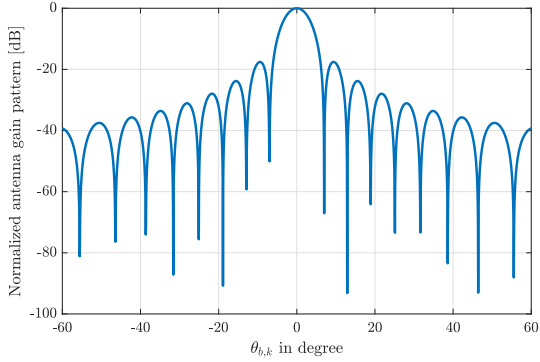} \vspace{-0.3cm}
	\caption{Normalized antenna gain pattern at the satellite with $\beta = 5 \lambda$.} \label{Fig:NBeamPattern}
			\vspace*{-0.45cm}
\end{figure}
\vspace*{-0.2cm}
	\begin{equation} \label{eq:gbk}
		g_{b,k} = g^{\max}g(\theta_{b,k}), \forall k \in \mathcal{K}_b, b \in \mathcal{B},
		\vspace*{-0.1cm}
	\end{equation}
	where $g^{\max}$ is the maximum beam pattern gain that is obtained if $\UEk$ is located at the beam center. The normalized antenna gain pattern  $g(\theta_{b,k}), \forall k \in \mathcal{K}_b, b \in \mathcal{B},$ is computed as
	\vspace*{-0.2cm}
	\begin{equation}\label{eq:gbktheta}
		g(\theta_{b,k}) =  \begin{cases}
			4\left|\frac{J_1\big(\frac{2\pi}{\lambda}\beta\sin(\theta_{b,k})\big)}{\frac{2\pi}{\lambda}\beta\sin (\theta_{b,k})}\right|^2 ,& \mbox{if } 0 < \theta_{b,k}\leq \frac{\pi}{2}, \\
			1, &\mbox{if } \ \theta_{b,k} = 0,
		\end{cases}
		\vspace*{-0.1cm}
	\end{equation}
	where $\beta$ is the radius of the antenna’s circular aperture, $\lambda$ is the carrier wavelength, and $J_1(\cdot)$ is the first kind Bessel function of the  order one. {\color{black} The Doppler shifts caused by the satellite mobility is assumed to be perfectly compensated by using proper estimation techniques \cite{8793224}.} %4284064
	 Fig.~\ref{Fig:NBeamPattern} shows an example of the normalized beam pattern with the aperture radius $\beta = 5 \lambda$. It explicitly unveils that the main energy is concentrated on the main lobe  of each beam, which gives  hints to locate the satellite beams hereafter.    
	After that, the statistical channel gain (SCG) at an arbitrary $\UEk$, denoted by $G_{b,k}$, is defined as
	\vspace*{-0.2cm}
	\begin{equation}
		G_{b,k} = {G_{\text{rx},b,k}g_{b,k}}/{L_{\text{fs},b,k}L_\text{atm}}, \forall k \in \mathcal{K}_b, b \in \mathcal{B},
		\vspace*{-0.1cm}
	\end{equation}
	where $L_{\text{fs},b,k}$  denotes the free space path loss, which is
	\vspace*{-0.2cm}
	\begin{equation}
		L_{\text{fs},b,k} = {16\pi^2 S_{b,k}^2}/{\lambda^2}, \forall k \in \mathcal{K}_b, b \in \mathcal{B},
		\vspace*{-0.1cm}
	\end{equation} 
	where $S_{b,k}$ denotes the slant range (line-of-sight distance along a slant direction) between the satellite and $\UEk$ in the $b$-th beam.  The received antenna power gain $G_{\text{rx},b,k}$ is \cite{3gpp2018study}
	\vspace*{-0.2cm}
	\begin{equation}\label{eq:Grxbk}
		G_{\mathrm{rx},b,k} =  {\epsilon_{b,k}\pi^2 d^2}/{\lambda^2}, \forall k \in \mathcal{K}_b, b \in \mathcal{B},
		\vspace*{-0.1cm}
	\end{equation}
	where $\epsilon_{b,k}$ is the antenna efficiency at $\UEk$ in the $b$-th beam and $d$ is the satellite antenna diameter.  We assume the network is in a noise-limited scenario, where a proper bandwidth allocation has been performed to avoid harmful levels of inter-beam interference\cite{8064672}. The received carrier-to-noise ratio (CNR) at $\UE_{k}$ served by beam $b$ is %calculated as
%	\vspace*{-0.2cm}
%\begin{equation}\label{}
$	\text{CNR}_k = {p_{b,k}G_{b,k}}/{\sigma^2_k}$,
%	\vspace*{-0.1cm}
%\end{equation}
where $\sigma_k^2$ denoted the noise variance at $\UEk$ and $p_{b,k}$ is the transmit power from $b$-th beam to $\UEk$.
		In the $b$-th satellite beam, the total gain at $\UEk$ through the space environment is a complicated expression of the practical aspects in multi-beam satellite communications comprising the beam pattern, effective transmitted and received antenna gains, free space path loss, and atmospheric loss. Several parameters, such as  $\tilde{g}(\theta_{b,k})$ defined in \eqref{eq:gbktheta} and $G_{\mathrm{rx},b,k}$ defined in \eqref{eq:Grxbk},  display the influences of hardware configurations from a specific reflector antenna with a circular aperture. Besides, the environmental parameters depending on the satellite beam locations, which are mathematically formulated in \eqref{eq:gbk}-\eqref{eq:Grxbk}, are of interest to design for the good link budget.  
%	\end{remark}
	\vspace*{-0.2cm}
	\section{Joint Beam Placement and Load Balancing Optimization}
	\vspace*{-0.2cm}
	This section formulates and solves a beam placement problem balancing  the active beams and the effective gains.
	\vspace*{-0.2cm}
	\subsection{Problem Statement} \label{Subsec:PS}
	\vspace*{-0.1cm}
	
	As analyzed in the previous section, the statistical channel gains are designable parameters, which are obtained by optimizing the beam centers. For such, let us introduce  $\bar{\mathbf{o}}_b = [\phi_b, \theta_b]^T  $ and $ \mathbf{o}_{b,k} = [\phi_{b,k},\theta_{b,k}]^T $ as the coordinate of the $b$-th beam's center and its served $\UE_{k}$, respectively. Here $\phi$ and $\theta$ stand for the corresponding longitude and latitude. We assume that
	all considered positions have the same elevation (e.g., set to zero for the sake of simplicity). Let us denote $ \tilde{d}( \bo_{b,k},\bar{\mathbf{o}}_b)$ the geography distance between the two coordinates, which is computed by the spherical law of cosines as 
%	\vspace*{-0.2cm}
%	\begin{equation}
		$ \tilde{d}( \bo_{b,k},\bar{\bo}_{b}) = R  \arccos(\vartheta_{b,k}), $
%		\vspace*{-0.1cm}
%	\end{equation}
	where $R$  is the Earth's radius and $\vartheta_{b,k}$ is defined as follows
	%\vspace*{-0.2cm}
	%\begin{equation}
		$\vartheta_{b,k} = \sin(\phi_{b,k})\sin(\phi_{b}) + \cos(\phi_{b,k})\cos(\phi_{b})\cos(\Delta\theta_{b,k} ),$
		%\vspace*{-0.1cm}
	%\end{equation}
	with $\Delta\theta_{b,k} = \theta_{b,k}-\theta_{b}$, $\forall k \in \mathcal{K}_b$ and $b \in \mathcal{B}$. We aim at minimizing a utility function that balances the number of beams operating at the satellite while attempting to push all users close to the beam centers as follows:
	\vspace*{-0.2cm}
	\begin{subequations} \label{probGlobal}
		\begin{alignat}{4}
			&\underset{ \{\mathcal{K}_b \}, \{\bar{\bo}_b \}}{\mathrm{minimize}}&\quad & w_1   \sum\nolimits_{b\in\Bcal}^{}  \sum\nolimits_{k \in \mathcal{K}_b } \tilde{d}( \mathbf{o}_{b,k} ,\bar{\bo}_b)^2 + w_2 |\Bcal|\label{probGlobala}\\
			&\mbox{subject to} && g_{b, k} \geq g^{\max}/2, \forall b\in\Bcal, k \in\Kcal_b, \label{probGlobalb}\\
			&&& \mathcal{K}_b \neq \emptyset,  \mathcal{K}_b  \subseteq \mathcal{K}, \forall b \in\Bcal, \label{probGlobalc} \\
			&&& \mathcal{K}_b \cap  \mathcal{K}_{b'} = \emptyset, \forall b, b'\in \Bcal, \label{probGlobald}\\
			&&&  \bigcup\nolimits_{b=1}^{|\mathcal{B}|} \mathcal{K}_b = \mathcal{K},\label{probGlobale}
			%			&&& \alpha_{b,k} \in \{0,1\}, \ \sum_{k=1}^{K}\alpha_{b,k} = 1, \forall b\in\Bcal,
		\end{alignat}
	\end{subequations}
	where the non-negative weights $w_1, w_2 \geq 0$ satisfy $w_1+w_2=1$, and respectively represent the different priorities to the objective function of problem~\eqref{probGlobal}. The weights $w_1$ and $w_2$ are flexibly designed to handle the conflicting metrics. Specifically, we stress that the former ensures that each user is located as close to its satellite beam center as possible, while minimizing the number of beams allows allocating more power to each satellite beam conditioned on the fixed power budget. The constraints \eqref{probGlobalb} ensure that each user is located in at least the half-power-bandwidth of its' served beam, which is motivated by Fig.~\ref{Fig:NBeamPattern}, where users should be laid in the main lobe of the beam. The constraints \eqref{probGlobalc} and \eqref{probGlobald} imply that each user is only served by a single beam, and each active beam must involve at least one user for energy efficiency conditioned on the limited power budget at the satellite. Furthermore, constraint \eqref{probGlobale} guarantees that all users are in the coverage area of the satellite. Problem~\eqref{probGlobal} is a mixed-integer non-convex program due to the inherent non-convexity of the objective function and the constraints \eqref{probGlobala}. A discrete feasible domain  \eqref{probGlobalc}-\eqref{probGlobale} makes problem~\eqref{probGlobal}  NP-hard and it may require an extremely high cost to find the global optimum. Hence, a heuristic algorithm can obtain an efficient sub-optimal solution with finite time consumption. {\hili The combinatorial structure might result in multiple solutions to \eqref{probGlobal}, so the unbalanced load issue can be mitigated by a solution that mostly spreads out users across the active  beams.}
	\vspace*{-0.2cm}	
	\subsection{Define the Number of Active Beams}
	\vspace*{-0.1cm}
	To shed the light on finding a minimal number of active beams, we	 consider a circular pattern of the $b$-th beam as shown in Fig.~\ref{Fig_example}(a) with the half power bandwidth defined by the angle $2\bar{\psi}, \forall b \in \mathcal{B}$. Based on the constraints \eqref{probGlobalb},  the $b$-th beam will serve users in its HPBW and therefore, the served user set  $\Kcal_b$ can be defined by verifying the angles among all the users. In more detail, the angle $\psi_{k,\ell}$ between two users $k$ and $\ell$, with $k, \ell, \in \Kcal_b$, seen from the satellite should hold that $\psi_{k,\ell} \leq \bar{\psi}, \ k,\ell\in\Kcal_b,\forall b\in\Bcal$. {\color{black} Using HPBW $2\bar{\psi}$ as in \cite{pachler2021static, alinque20a} may result in some users not to be covered. Instead, we use  the quantity $\bar{\psi}$ to guarantee that all users in the same beam are in the HPBW area of the satellite's beams.} Therefore, to define the minimal number of active beams, we propose to solve the following optimization problem
	\vspace*{-0.15cm}
	\begin{subequations} \label{prob_beam}
		\begin{alignat}{4}
			&\underset{\{\mathcal{K}_b\}}{\mathrm{minimize}}&\quad &  |\Bcal|\label{prob_beama}\\
			&\mbox{subject to} && \psi_{k,\ell} \leq \bar{\psi}, \ \forall k,\ell\in\Kcal_b,\forall b\in\Bcal, \label{prob_beamb}\\
			&&& \eqref{probGlobalc}-\eqref{probGlobale}.
		\end{alignat}
	\end{subequations}

	\begin{figure}[t]
		\begin{minipage}{0.2\textwidth}
			\centering
			\includegraphics[trim=0.0cm 0.0cm 0.0cm 0.5cm, clip=true, width=1.05in]{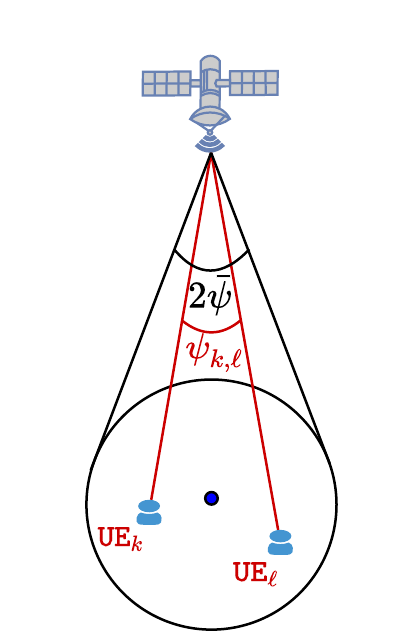} \\
			$(a)$
			\vspace*{-0.1cm}
		\end{minipage}
		\begin{minipage}{0.2\textwidth}
			\centering
			\includegraphics[trim=1.0cm 0.2cm 1.4cm 0.8cm, clip=true, width=2.0in]{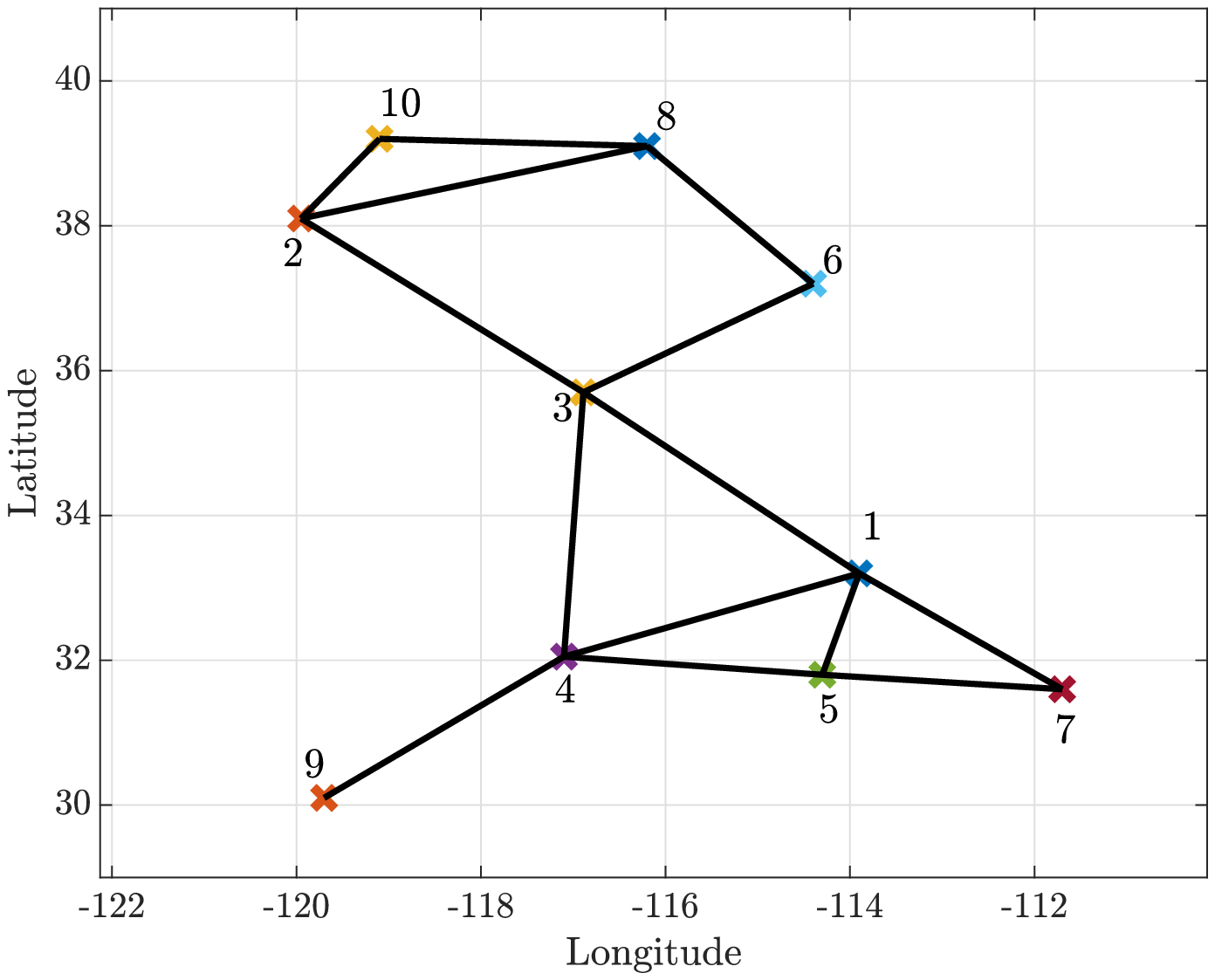} \\ \vspace*{-0.1cm}
			$(b)$
			\vspace*{-0.1cm}
		\end{minipage}
		\caption{ The considered satellite system:
			$(a)$ Example of one beam pattern with two users in the beam; $(b)$ Graph $G$ with ten vertices and $14$ edges.} \label{Fig_example}
		\vspace*{-0.6cm}
	\end{figure}
	Despite the simplification compared with the original problem~\eqref{probGlobal}, the non-convex and non-smooth properties are still preserved in problem~\eqref{prob_beam}. To effectively find a  good feasible solution, we develop a heuristic algorithm for problem \eqref{prob_beam} based on the graph theory. In particular, a graph $G$ is introduced with a pair $(\mathcal{V}(G), \mathcal{E}(G))$. The set $\mathcal{V}(G) =\{ 1, \ldots, K \}$ consists of the $K$ vertices each related to a user. Meanwhile, the set $\mathcal{E}(G)$ of edges in which each edge represents a connection between two distinct vertices, i.e.,
	\vspace*{-0.2cm}
	\begin{equation} \label{eq:Edge}
		\mathcal{E}(G) = \{ \{k,\ell\}| \psi_{k,\ell} \leq \bar{\psi}, \forall k,\ell\in \mathcal{V}(G) \},
		\vspace*{-0.1cm}
	\end{equation}
	and noticing that no two vertices are connected by more than one edge.\footnote{In our considered framework, $G$ should be an undirected graph, that is, $\{k,\ell\} \in \mathcal{E}(G)$ is equivalent to $\{l,k\} \in \mathcal{E}(G)$.}  From the constraints \eqref{prob_beamb}, we can construct the adjacency matrix $\bU \in \mathbb{R}^{K\times K}$ of graph $G$, whose $(k,l)$-th element, denoted by $[\bU]_{k,\ell}$,  is defined as follows
	\vspace*{-0.2cm}
	\begin{equation}\label{eq_U}
		[\bU]_{k,\ell} = \begin{cases}
			1, & \mbox{if } \psi_{k,\ell} \leq \bar{\psi}, k,\ell\in\Kcal,\\
			0, & \mbox{otherwise}.
		\end{cases}
		\vspace*{-0.1cm}
	\end{equation}
	From \eqref{eq_U}, the two vertices $k$ and $\ell$ are adjacent if $[\bU]_{k,\ell} =1$ meaning that there is an edge between them. One feature of graph~$G$ can be observed in Theorem~\ref{theorem1}.
	\begin{theorem} \label{theorem1}
		From the HPBW criterion defining the edges in \eqref{eq:Edge}, if each active beam serves at least $\lceil (K-1)/2 \rceil$ users  where $\lceil \cdot \rceil$ is the ceiling function,  the graph $G$ is connected.\footnote{A connected graph always exists a path between two arbitrary vertices.} The number of active beams is lower bounded by
		%\vspace*{-0.2cm}
		%\begin{equation}
			$|\mathcal{B}| \geq \left\lceil K/\lceil (K-1)/2 \rceil \right\rceil.$
		%	\vspace*{-0.1cm}
		%\end{equation}
	\end{theorem}
	{\color{black}\begin{proof}
	The proof is to show that the neighborhoods of vertices are overlapping and is omitted due to space limitations.
\end{proof}}
%\vspace*{-0.1cm}
	Even though the graph $G$ may be disconnected in practice, this theorem establishes a necessary condition so that each active beam can serve multiple users when they are not quite distant from each other based on the HPBW. We can provide a lower bound on the number of active beams as the graph $G$ is connected.  The network can place the beam centers effectively by the  following important property.
	\begin{lemma} \label{lemma:clique}
		If multiple users belong to the same beam, they formulate a clique in which the vertices are mutually adjacent.
	\end{lemma}
	\begin{proof}
		The proof is accomplished by using the properties in \eqref{eq:Edge} and \eqref{eq_U}, which is omitted due to space limitations. 
	\end{proof}
	In Lemma~\ref{lemma:clique}, each clique is a subset of vertices that usually consists of more than three users. Subsets including one or two users per beam are also called  cliques for the sake of convenience. It provides an effective hint to place the beam center  based on the half power bandwidth, which  links to a clique of the graph~$G$.  The main idea to solve problem~\eqref{prob_beam} is that if each clique contains more vertices, the less number of active beams needs to serve all the users.  
	
	Our proposed approach is summarized in the first stage of Algorithm~\ref{alg_Find_B2}. It starts with formulating  the graph~$G$ with the sets of vertices   $\mathcal{V}(G)$ and edges $\mathcal{E}(G)$  for the satellite and user's locations by using \eqref{eq:Edge}. After that,  we can compute the adjacency matrix $\mathbf{U} \in \mathbb{R}^{K \times K}$  as shown in \eqref{eq_U}. Let us denote $\mathcal{H}_i = \{ \tilde{\Xcal}_{b,i} \}_{b\in\Bcal}$ the subset comprising of cliques $\tilde{\Xcal}_{b,i}$ with the same number of edges, e.g., $\Hcal_i$ includes the cliques with single edge and so on. Here, the index $i = 1, \ldots,i_{\max}$ stands for the number of edges in a clique. Here, $i_{\max}$ is the maximum number of edges in a clique. Next, we introduce $\tilde{\Kcal} = \{ \tilde{\Kcal}^{(\kappa)} \}$ the dictionary that contains all possibilities to the solution of  problem~\eqref{prob_beam}, each denoted as $ \tilde{\Kcal}^{(\kappa)}$. At the beginning,  $\tilde{\Kcal}$ is initially set to be empty. As aforementioned, the combinatorial structure may result in multiple solutions to problem~\eqref{prob_beam}. Consequently, Algorithm~\ref{alg_Find_B2} must scrutinize many available combinations in an iterative manner to find the minimum number of active beams to serve all the users. Alternatively, the number of solutions $\{\tilde{\Kcal}^{(\kappa)}\}$ will be gradually expanded along with iterations. The iterative approach starts scanning $\Hcal_{i_{\max}}$ to define  $\{ \tilde{\Kcal}^{(\kappa)} \}$, each formulated as
	\vspace*{-0.1cm}
	\begin{equation}
		\begin{split}
			\tilde{\Kcal}^{(\kappa)} =&  \big\{ \tilde{\Xcal}_{b,i_{\max}} \big|  \tilde{\Xcal}_{b,i_{\max}} \cap  \tilde{\Xcal}_{b',i_{\max}} = \emptyset, \\
			& \forall  \tilde{\Xcal}_{b,i_{\max}}, \tilde{\Xcal}_{b',i_{\max}} \subseteq \Hcal_{i_{\max}} \big\}, \forall \kappa.
		\end{split}
		\vspace*{-0.1cm}
	\end{equation}
	After that, the dictionary $\tilde{\Kcal}$ will be updated by scrutinizing all the remaining subsets $\{\Hcal_i\} \setminus \Hcal_{i_{\max}}$. The subset $\Hcal_{i_{\max} - i}$ is investigated at iteration~$i$ and the dictionary is expanded as
	\vspace*{-0.2cm}
	\begin{equation} \label{eq:Kkappa}
		\tilde{\Kcal}^{(\kappa)} \leftarrow \tilde{\Kcal}^{(\kappa)} \cup \{\tilde{\Xcal}_{b,i}^\ast\}, \forall \kappa,
		\vspace*{-0.1cm}
	\end{equation}
	where the following definition holds $ 		\tilde{\Xcal}_{b,i_{\max} - i}^\ast = \big\{ \tilde{\Xcal}_{b,i_{\max} -i} \big|
 \tilde{\Xcal}_{b,i_{\max}-i} \cap  \tilde{\Xcal}_{b',i_{\max}-i} = \emptyset,  \forall  \tilde{\Xcal}_{b,i_{\max}-i},
 \tilde{\Xcal}_{b',i_{\max}-i} \subseteq \Hcal_{i_{\max}-i},  \tilde{\Xcal}_{b,i_{\max} - i} \cap \tilde{\Kcal}^{(\kappa)} = \emptyset  \big\}. $
	Even though the dictionary update in \eqref{eq:Kkappa} may lead to a local solution to problem~\eqref{prob_beam},  the condition $ \tilde{\Xcal}_{b,i_{\max} - i} \cap \tilde{\Kcal}^{(\kappa)} = \emptyset$  truncates many cliques that makes our proposed approach have the total cost significantly lower than an exhaustive search. After scrutinizing through all  the subsets $\Hcal_i, \forall i,$ or all users allocated, the optimized solution $\{\Kcal_b\}$ is then selected from all the possible solutions in $\tilde{\Kcal}$ as
	\vspace*{-0.2cm}
	\begin{equation}\label{eq_updateKb}
		\{\Kcal_b \} = \underset{\tilde{\Kcal}^{(\kappa)}\subset\tilde{\Kcal}}{\mathrm{argmin}} \ |\tilde{\Kcal}^{(\kappa)}|.
		\vspace*{-0.1cm}
	\end{equation}
	We notice that once $\{\Kcal_b \}$ is obtained, the beam center of the $b$-th beam, $\forall b$, can be computed as
	\vspace*{-0.2cm}
	\begin{equation}\label{eq_update_beam_center}
		\bar{\bo}_b = \frac{1}{|\Kcal_b|}\sum\nolimits_{k=1}^{|\Kcal_b |}\bo_{b,k}, \forall b\in\Bcal.
		\vspace*{-0.1cm}
	\end{equation}
	For the sake of the clarity, one toy example presenting a realization of users' location for a multi-beam satellite system serving $10$ users is given  as follows.
	
	\vspace*{-0.2cm}
	\begin{example}
		We consider a MEO satellite system where the parameter setting is given in Section~\ref{Sec:NumericalResults}. On the ground, there are $10$ users with their locations illustrated in Fig. \ref{Fig_example}(b). Users whom the same beam can serve are represented by a link between the two corresponding locations. By utilizing \eqref{eq:Edge} with $2\bar{\psi} = 3.2^\circ$ \cite{hpbw}, we can construct a graph $G $ as shown in Fig.~\ref{Fig_example}(b) comprising $\mathcal{V}(G) = \{1, \ldots, 10 \}$ and $\mathcal{E}(G) = \{ \{k,\ell\}| \psi_{k,\ell} \leq 1.6^\circ, \forall k,\ell\in \mathcal{V}(G) \}$. Subsequently, 
		the adjacency matrix is formulated as 
		\vspace*{-0.2cm}
		\begin{equation} \label{eq:U10}
			\mathbf{U}= \begin{bmatrix}
				1 & 0 & 1 & 1 & 1 & 0 & 1 & 0 & 0 & 0\\
				0 & 1 & 1 & 0 & 0 & 0 & 0 & 1 & 0 & 1\\
				1 & 1 & 1 & 1 & 0 & 1 & 0 & 0 & 0 & 0\\
				1 & 0 & 1 & 1 & 1 & 0 & 0 & 0 & 1 & 0\\
				1 & 0 & 0 & 1 & 1 & 0 & 1 & 0 & 0 & 0\\
				0 & 0 & 1 & 0 & 0 & 1 & 0 & 1 & 0 & 0\\
				1 & 0 & 0 & 0 & 1 & 0 & 1 & 0 & 0 & 0\\
				0 & 1 & 0 & 0 & 0 & 1 & 0 & 1 & 0 & 1\\
				0 & 0 & 0 & 0 & 0 & 0 & 0 & 0 & 1 & 0\\
				0 & 1 & 0 & 0 & 0 & 0 & 0 & 1 & 0 & 1\\
			\end{bmatrix}.
			\vspace*{-0.1cm}
		\end{equation}
		Utilizing the adjacency matrix in \eqref{eq:U10} and the undirected property, all possible cliques can be synthesized as
		\vspace*{-0.2cm}
		\begin{align}
			\Hcal_1=& \{(1), (2), (3), (4) ,(5), (6), (7), (8), (9), (10)\},\\
			\Hcal_2 =& \{(1,3),	(1,4),	(1,5),	(1,7),	(2,3),	(2,8), (2,10), \label{eq:Clique2}\\	&	(3,4),	(3,6),	(4,5),	(4,9),	(5,7),	(6,8),	(8,10)\}, \notag \\
			\Hcal_3 =& \{(1,3,4),	(1,4,5),	(1,5,7),	(2,8,10)\}, \label{eq:Clique3}
			\vspace*{-0.1cm}
		\end{align}
		From the cliques with the $3$ edges involved in $\Hcal_3$, the following cliques can coexist
		\vspace*{-0.2cm}
		\begin{align}
			\tilde{\Kcal}^{(1)} &= \{(2,8,10),(1,3,4)\}, \\
			\tilde{\Kcal}^{(2)} &= \{(2,8,10),(1,4,5)\}, \\
			\tilde{\Kcal}^{(3)} &= \{(2,8,10),(1,5,7)\},
			\vspace*{-0.1cm}
		\end{align}
		which indicate that the cliques should be non-overlapping to fulfill the constraints \eqref{probGlobald}. We further expand  $\tilde{\Kcal}^{(\kappa)}, \kappa \in \{ 1,2,3\}$ by adding the  cliques with the $2$ edges in \eqref{eq:Clique2} and update  $\{\tilde{\Kcal}^{(\kappa)}\}$ to obtain the solutions. More specifically, $\tilde{\Kcal}^{(1)} \leftarrow \tilde{\Kcal}^{(1)} \cup \{(5,7)\}$, $\tilde{\Kcal}^{(2)} \leftarrow \tilde{\Kcal}^{(2)} \cup \{(3,6)\} $, and $\tilde{\Kcal}^{(2)} \leftarrow \tilde{\Kcal}^{(2)} \cup \{(3,6), (4,9)\} $, which result in
		\vspace*{-0.2cm}
		\begin{align}
			\tilde{\Kcal}^{(1)} &= \{(2,8,10),(1,3,4), (5,7)\},\\
			\tilde{\Kcal}^{(2)} &= \{(2,8,10),(1,4,5), (3,6)\}, \\
			\tilde{\Kcal}^{(3)} &= \{(2,8,10),(1,5,7), (3,6), (4,9)\},
			\vspace*{-0.1cm}
		\end{align}
		Finally, $\tilde{\Kcal}^{(\kappa)}, \kappa \in \{1,2, 3 \}$,  are expanded by utilizing $\Hcal_1$ as $\tilde{\Kcal}^{(1)} \leftarrow \tilde{\Kcal}^{(1)}\cup \{(6), (9)\}$, $\tilde{\Kcal}^{(2)} \leftarrow \tilde{\Kcal}^{(2)}\cup \{(7),(9)\}$, and $\tilde{\Kcal}^{(3)} \leftarrow \tilde{\Kcal}^{(3)}$, which result in
		\vspace*{-0.2cm}
		\begin{align}
			\tilde{\Kcal}^{(1)} &= \{(2,8,10),(1,3,4), (5,7), (6), (9)\}, \\
			\tilde{\Kcal}^{(2)} &= \{(2,8,10),(1,4,5), (3,6), (7), (9)\}, \\
			\tilde{\Kcal}^{(3)} &= \{(2,8,10),(1,5,7), (3,6), (4,9)\}.
			\vspace*{-0.1cm}
		\end{align}
		Using \eqref{eq_updateKb}, the last solution including the minimum number of set, i.e., $\{\Kcal_b \} = \tilde{\Kcal}^{(3)} = \{(2,8,10),(1,5,7), (3,6), (4,9)\}$. Therefore, the number of active beams  should be $ |\Bcal| = 4$.
	\end{example}
	\begin{algorithm}[t]
	\begin{algorithmic}[1]\fontsize{10}{10.5}\selectfont
		\protect\caption{A two-stage algorithm to solve problem \eqref{probGlobal}} % \eqref{probGlobal}}
		\label{alg_Find_B2}
		\global\long\def\algorithmicrequire{\textbf{INPUT:}}
		\REQUIRE The longitudes and latitudes of users and satellites, $\bo_{b,k}, \bar{\bo}_b, \forall b,k$; the HPBW $2{\bar{\psi}}$; 
		\STATE \% \textit{Stage 1: Define the number of active beams}
		\STATE Formulate graph $G$ with a pair $(\mathcal{V}(G), \mathcal{E}(G))$ and compute the adjacency matrix $\bU$ as in \eqref{eq_U}.
		\STATE Compute all possible cliques $\mathcal{H}_i = \{ \tilde{\Xcal}_{b,i} \}$ from graph $G$ by utilizing Lemma~\ref{lemma:clique}.
		\STATE Initial $\{\tilde{\Kcal}^{(0)}\}$%$\tilde{\Kcal} = \{\tilde{\Kcal}_{\kappa}\}$
		\FOR{each group set $\Hcal_i$ in $\Hcal$}
		%		\STATE Compute $\Xcal = \{\Xcal_j\}$ as in \eqref{eq_condition_X}
		%		\STATE Set $\kappa=1$, initial cluster $\bar{\Kcal}$
		\FOR{each group set $\tilde{\Kcal}^{(\kappa)}$ in  $\{\tilde{\Kcal}^{(\kappa')}\}$}
		\IF{$|\tilde{\Kcal}^{(\kappa)}| < K$}
		\STATE Update $\tilde{\Kcal}^{(\kappa)}$ as in \eqref{eq:Kkappa}.
		\ENDIF
		\ENDFOR
		%		\STATE Update $\tilde{\Kcal} = \bar{\Kcal}$
		\ENDFOR
		\STATE (Optional) Update $\{\Kcal_b^*\}$ as in \eqref{eq_updateKb} and beam centers as in \eqref{eq_update_beam_center}.
		\STATE \% \textit{Stage 2: Refine the beam centers}
		%		\WHILE {$ \psi_{k,\ell} > \bar{\psi}, \ \forall k,\ell\in\Kcal_b,\forall b\in\Bcal$}
		\STATE Transform all user locations $\{\bo_k\}$ to  the Cartesian coordinates  by utilizing $x = R\cos(\phi)\cos(\theta)$, $y = R \cos(\phi) \sin(\theta)$, and $z = R \sin(\phi)$.
		%		\STATE	Perform the K-means algorithm to problem \eqref{prob_Kmean_Cartes}.
		%		\WHILE{any $\mathbf{U}_{k,\ell} \neq 1, \forall k,k,\ell\in\Kcal_b $}
		%		\STATE	Perform the K-means algorithm to problem \eqref{prob_Kmean_Cartes}.
		%		\ENDWHILE
		\STATE \textbf{do}
		\STATE \hspace*{6.5pt}	Perform the K-means clustering  to problem~\eqref{prob_Kmean_Cartesv1}.
		\STATE \textbf{while} any $\mathbf{U}_{k,\ell} \neq 1, \forall k,k,\ell\in\Kcal_b $.
		\STATE Transform all  the beam centers $\{\bmu_b\}$ to  the geographic coordinates as $\phi = \arcsin(z/R)$ and $\theta = \arctan(y/x)$.
		%		\ENDWHILE
		\global\long\def\algorithmicrequire{\textbf{OUTPUT:}}
		\REQUIRE The user sets $\{\Kcal_b^*\}$ and the number of active beams $B = |\tilde{\Kcal}|$.
	\end{algorithmic}
%	\vspace*{-1cm} 
\end{algorithm}
	\vspace*{-0.35cm}
	\subsection{Refine the solution with K-means clustering}
	\vspace*{-0.15cm}
	By applying the first stage, the number of active beams has been obtained with a guarantee that users are in the HPBW of each beam. The main drawback of this stage is that this stage has focused on maximizing the number of users in one beam, which might cause an unbalanced load among the beams. As a result, some active beams need to serve many users, while the remaining is in a sparse situation with a few users. {\hili For a fairness level with a balanced load among the active beams, our goal is to minimize the distance from every user to its beam center targeting at a homogeneous network.} Mathematically, we solve the following optimization problem
	%	\begin{subequations} \label{prob_Kmean}
	%		\begin{alignat}{4}
	%			&\underset{ \{\pmb{\mu}_b \}, \{\alpha_{b,k}\}}{\mathrm{minimize}}&\quad &    \sum_{b =1}^{|\mathcal{B}|}  \sum_{k =1  }^{|\mathcal{K}|}\alpha_{b,k}\left( \tilde{d}( \mathbf{o}_{k} ,\pmb{\mu}_b)\right)^2  \label{prob_Kmeana}\\
	%			&\mbox{subject to} &&  \alpha_{b,k} \in \{0,1\}, \ \forall b\in\Bcal, k\in\Kcal \label{prob_Kmeanb}\\
	%			&&& \sum_{k=1}^{|\Kcal|}\alpha_{b,k} = 1, \forall b\in\Bcal,\label{prob_Kmeanc}
	%		\end{alignat}
	%	\end{subequations}
	\vspace*{-0.2cm}
	\begin{subequations} \label{Problem:KmeanV1}
		\begin{alignat}{4}
			&\underset{\{\Kcal_b\}, \{\bar{\bo}_b \}}{\mathrm{minimize}}&\quad &    \sum\nolimits_{b \in\Bcal}^{}  \sum\nolimits_{k \in \mathcal{K}_b } \tilde{d}( \mathbf{o}_{b,k} ,\bar{\bo}_b)^2 \label{Problem:KmeanV1a}\\
			&\mbox{subject to} && \eqref{probGlobalb}-\eqref{probGlobale}.
		\end{alignat}
	\end{subequations}
%	\vspace*{-0.1cm}
	%	where $\alpha_{b,k}$ is the introduced slack binary variables indicating the relationship between $\UE_{k}$ and beam $b$. Specifically, $\alpha_{b,k}=1$ corresponds to that $\UE_{k}$ is served by $b$-th beam. Otherwise, $\alpha_{b,k}=0$. 
	%	\begin{figure}[t]
	%		\begin{minipage}{0.24\textwidth}
	%			\centering
	%			\includegraphics[width= 0.8 \textwidth]{figures/alg1_beam.PNG} \\
	%			$(a)$
	%			\vspace*{-0.3cm}
	%		\end{minipage}
	%		\hfill
	%		\begin{minipage}{0.24\textwidth}
	%			\centering
	%			\includegraphics[width= 0.8 \textwidth]{figures/kmean_beam.PNG} \\
	%			$(b)$
	%			%		\vspace*{-0.1cm}
	%		\end{minipage}
	%		\caption{ The beam clustering by using
	%			$(a)$ algorithm 1; $(b)$ K-means.} \label{Fig_kmean}
	%		%	\vspace*{-15pt}
	%	\end{figure}
	%	It can be observed that the problem $\eqref{prob_Kmean}$ is not convex with respect to variable sets $\{\pmb{\mu}_b \}$ and $ \{\alpha_{b,k}\}$. 
	In order for the network to handle problem~\eqref{Problem:KmeanV1}, a special mechanism should be applied for the geographic coordinates. In particular, we convert each geographic coordinate denoted by the longitude and latitude $(\phi, \theta)$ to the corresponding Cartesian coordinate denoted by $(x,y,z)$ by exploiting the following relationship $x = R\cos(\phi)\cos(\theta)$, $y = R \cos(\phi) \sin(\theta)$, and $z = R \sin(\phi)$.
	%\begin{equation}\label{latlon2local}
	%	\begin{cases}
	%	x &= R\cos(\phi)\cos(\theta),  \\
	%	y &= R \cos(\phi) \sin(\theta),  \\
	%	z &= R \sin(\phi).
	%\end{cases}
	%\end{equation}
	Hence,  the Cartersian coordinates $\tilde{\bo}_{b,k}$ and $\tilde{\bar{\bo}}_b$ are obtained for all the users and the beam center of the $b$-th beam by utilizing  $\bo_{b,k}$ and $\bar{\bo}_b, \forall b,k$, respectively. Next, problem~\eqref{Problem:KmeanV1} is reformulated as
	\vspace*{-0.2cm}
	\begin{subequations} \label{prob_Kmean_Cartes}
		\begin{alignat}{4}
			&\underset{\{\Kcal_b\}, \{\tilde{\bar{\bo}}_b \}}{\mathrm{minimize}}&\quad &    \sum\nolimits_{b \in\Bcal}^{}  \sum\nolimits_{k \in \mathcal{K}_b } \|\tilde{\bo}_{b,k} - \tilde{\bar{\bo}}_b\|^2 \label{}\\
			&\mbox{subject to} && \eqref{probGlobalb}-\eqref{probGlobale},
		\end{alignat}
	\end{subequations}
	%	\begin{subequations} \label{prob_Kmean_Cartes}
	%		\begin{alignat}{4}
	%			&\underset{ \{\pmb{\tilde{\mu}}_b \}, \{\alpha_{b,k}\}}{\mathrm{minimize}}&\quad &    \sum_{b =1}^{|\mathcal{B}|}  \sum_{k =1  }^{|\mathcal{K}|}\alpha_{b,k}|| \tilde{\bo}_k - \tilde{\bmu}_b ||^2  \label{prob_Kmean_Cartesa}\\
	%			&\mbox{subject to} &&  \eqref{prob_Kmeanb}, \eqref{prob_Kmeanc}.
	%		\end{alignat}
	%	\end{subequations}
	where $\|\cdot\|$ is the Euclidean norm. By neglecting the constraints \eqref{probGlobalb},  problem~\eqref{prob_Kmean_Cartes} is aligned with the standard form of the K-means clustering as
	\vspace*{-0.2cm}
	\begin{subequations} \label{prob_Kmean_Cartesv1}
		\begin{alignat}{4}
			&\underset{\{\Kcal_b\}, \{\tilde{\bar{\bo}}_b \}}{\mathrm{minimize}}&\quad &    \sum\nolimits_{b \in\Bcal}^{}  \sum\nolimits_{k \in \mathcal{K}_b } \|\tilde{\bo}_{b,k} - \tilde{\bar{\bo}}_b\|^2 \label{}\\
			&\mbox{subject to} && \eqref{probGlobalc}-\eqref{probGlobale},
			\vspace*{-0.0cm}
		\end{alignat}
	\end{subequations}
	and therefore the optimal solution to $\{ \Kcal_b \}$ and $\{\tilde{\bar{\bo}}_b \}$ can be obtained in polynomial time. Let us denote  $\{ \tilde{\bar{\bo}}_b^*\}$ the optimal solution to the beam centers, we can reverse the optimal geographic coordinates. Noting that  a Cartesian coordinate $(x,y,z)$, the corresponding geographic coordinate $(\phi,\theta)$ is mathematically computed as $\phi = \arcsin(z/R)$ and $\theta = \arctan(y/x)$.
	%\begin{equation}\label{local2latlon}
	%\begin{cases}
	%	\phi &= \arcsin(z/R), \\
	%	\theta &= \arctan(y/x),
	%\end{cases}
	%\end{equation}
	Since the relaxed problem~\eqref{prob_Kmean_Cartesv1} does not guarantee the HPBW requirements, we may need to perform the K-means clustering approach several times. The beam centers and users per beam are refined based on the K-means clustering, which is summarized in the second stage as shown in Algorithm~\ref{alg_Find_B2}. The proposed algorithm combines the benefits of both the aforementioned stages by efficiently solving problem \eqref{probGlobal} in a hierarchical fashion.
	\vspace*{-0.2cm}
	\section{Numerical Results} \label{Sec:NumericalResults}
	\vspace*{-0.15cm}
	The considered beam placement framework is testified by a MEO satellite serving simultaneously multiple users randomly located in an area with latitude and longitude within the ranges as $[30, 40]$ and $[-120, -110]$, respectively. The satellite is located at the coordinate whose the [latitude, longitude]  is $[0^\circ, -88.7^\circ]$ and its  altitude is  $8063$~[km]. The total transmit power at the satellite is $23.5$~[dBW]. The carrier frequency is $18.05$~[GHz] and the aperture radius is $5\lambda$ with $\lambda$ being the wavelength. The satellite antenna diameter is $0.6$~[m], while the maximum gain at each beam center is $50$~[dBi] and HPBW $2\bar{\psi} = 3.2^\circ$ \cite{hpbw}. The noise variance is $-118$~[dBW]. %Numerical results are reported over $1000$ realizations of different users' locations.
%	\textbf{\begin{figure}[t]
%			\centering
%			\includegraphics[trim=0.3cm 0.0cm 1.2cm 0.4cm, clip=true, width=2.5in]{figures/FigSNR.eps}
%			\vspace*{-0.25cm}
%			\caption{CDF of the SCGNR [dB] per user.}
%			\label{fig:SNR}
%			\vspace*{-0.5cm}
%		\end{figure}
%		\begin{figure}[t]
%			\centering
%			\includegraphics[trim=0.3cm 0.0cm 1.2cm 0.4cm, clip=true, width=2.5in]{figures/Gap_between_beams.eps}
%			\vspace*{-0.25cm}
%			\caption{Load balancing gap versus the total number of users.}
%			\label{fig:gap}
%			\vspace*{-0.45cm}
%	\end{figure}}
\begin{figure*}[t]
	\begin{minipage}{0.5\textwidth}
		\centering
		\includegraphics[trim=0.3cm 0.0cm 1.2cm 0.4cm, clip=true, width=2.5in]{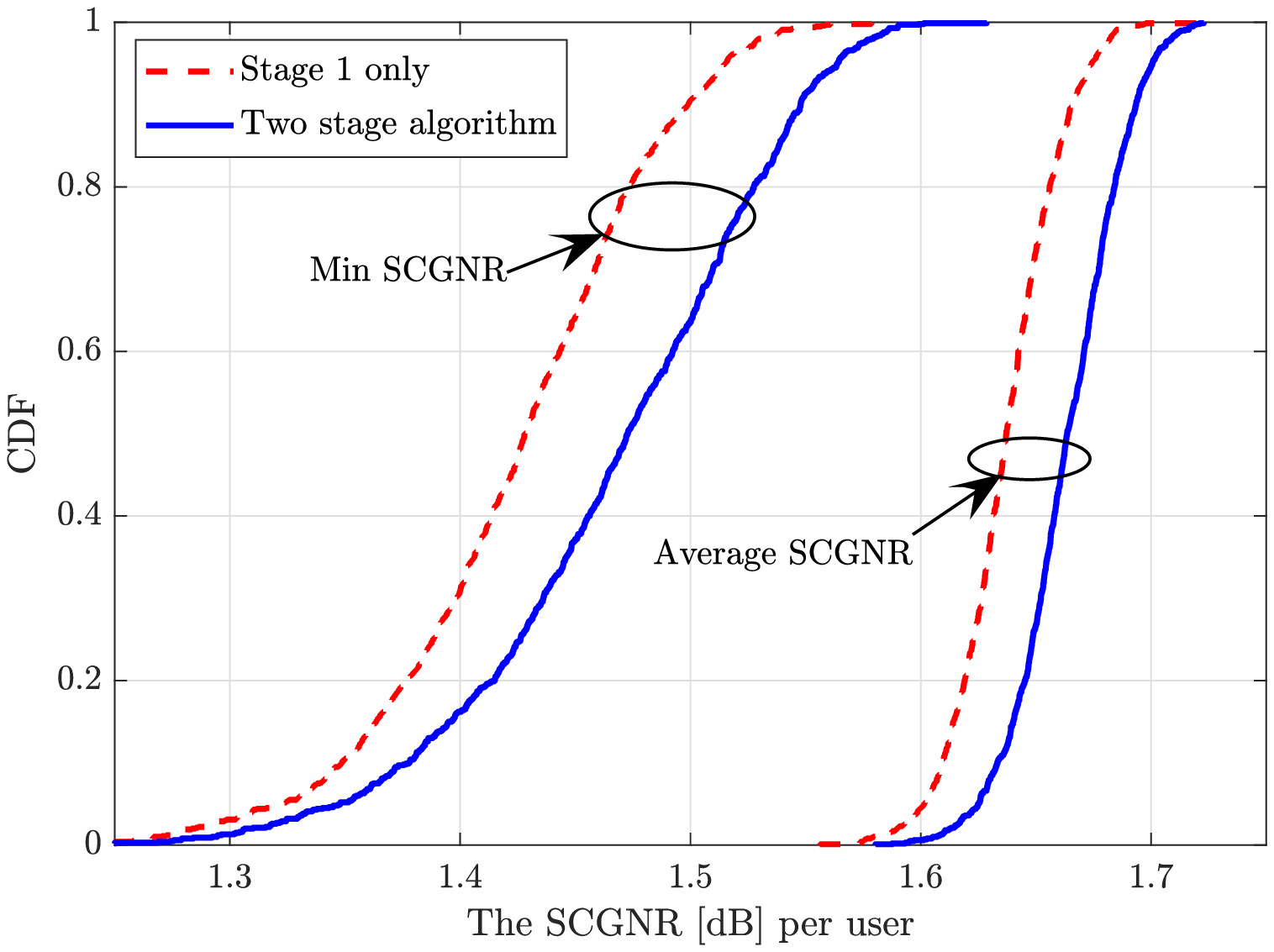} \\
		$(a)$
					\vspace*{-0.1cm}
%		\label{fig:SNR}
		%\vspace*{-0.6cm}
	\end{minipage}
	\hfill
	\begin{minipage}{0.5\textwidth}
		\centering
		\includegraphics[trim=0.3cm 0.0cm 1.2cm 0.4cm, clip=true, width=2.5in]{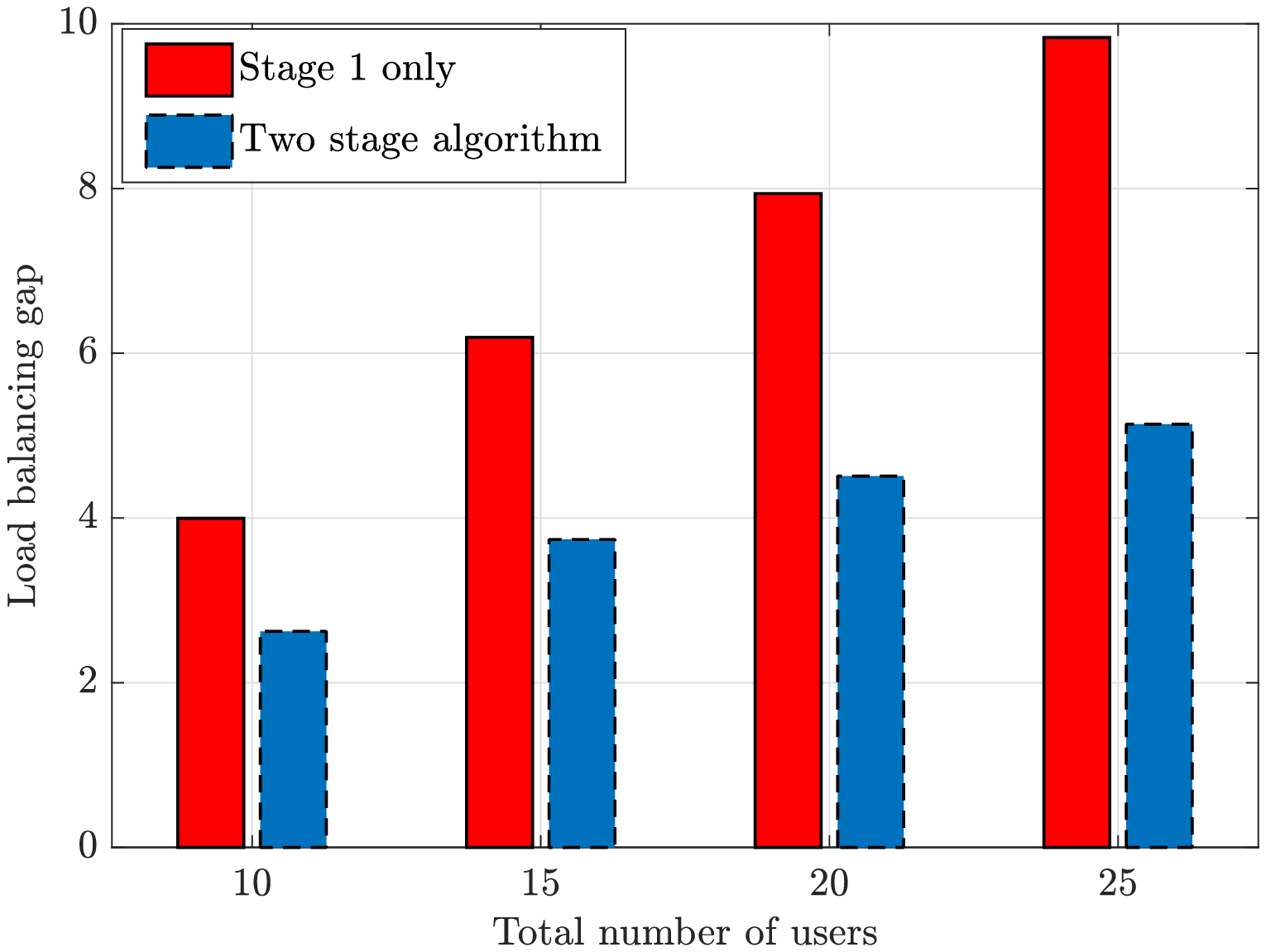} \\
		$(b)$
					\vspace*{-0.1cm}
		%\caption{The probability of satisfied users versus the %QoS requirement.}
%		\label{fig:gap}
		%\vspace*{-0.6cm}
	\end{minipage}
	\caption{\textcolor{black}{The system performance:
			$(a)$ plots CDF of the SCGNR [dB] per user; and $(b)$ plots load balancing gap versus the total number of users.}} \label{Fig:SysPer}
			\vspace*{-15pt}
			\label{fig:gap}
\end{figure*}
	\begin{table}[t]
		\caption{The CNR per user by utilizing  different benchmarks.}
		\label{SNR}
		\centering
		%		\vspace*{-0.25cm}
		\begin{tabular}{|ll|c|c|c|c|}
			\hline
			\multicolumn{2}{|l|}{\backslashbox{Ben./CNR}{No. of users}}                             & 10    & 15    & 20    & 25    \\ \hline
			\multicolumn{1}{|l|}{\multirow{2}{*}{Beam Aperture}}         & Min CNR [dB] & 18.98 & 18.15 & 17.54 & 17.09 \\ \cline{2-6} 
			\multicolumn{1}{|l|}{}                                     & Avg. CNR [dB] & 19.09 & 18.28 & 17.69 & 17.25 \\ \hline
			\multicolumn{1}{|l|}{\multirow{2}{*}{Homo. Balance }}          & Min CNR [dB]  & 19.97 & 19.08 & 18.64 & 18.27 \\ \cline{2-6} 
			\multicolumn{1}{|l|}{}                                     & Avg. CNR [dB] & 20.44 & 19.68 & 19.25 & 18.98 \\ \hline
			\multicolumn{1}{|l|}{\multirow{2}{*}{Algorithm~\ref{alg_Find_B2}}} & Min CNR [dB] & 20.36 & 19.58 & 19.12 & 18.85 \\ \cline{2-6} 
			\multicolumn{1}{|l|}{}                                     & Avg. CNR [dB] & 20.52 & 19.76 & 19.32 & 19.05 \\ \hline
		\end{tabular}
				\vspace*{-0.35cm}
	\end{table}
%	\vspace*{-0.05cm}

	Fig.~\ref{fig:gap}(a) considers a network with $20$~users and plot  the cumulative distribution function (CDF) of the the statistical channel gain to noise ratio (SCGNR) [dB] for each user, i.e., $10\log_{10}(G_{b,k}/\sigma^2_k), \forall b, k$. Both the average and minimum SCGNR of the users by exploiting Stage~$1$ only or the two-stage algorithm (Algorithm~\ref{alg_Find_B2}). The results show that the two stages bring benefits to both the average and min SCGNRs. Fig~\ref{fig:gap}(b) displays the load balancing gap between beams as a function of the number of users in the network. Stage~$2$ helps minimize the distance from users to its' beam centers while balancing the number of users among beams in the system, thereby offering a smaller load balancing gap than by only utilizing Stage~$1$. % The load balancing gap gets bigger as the number of available users increases.
	
	The link budget, i.e., the received carrier-to-noise ratio (CNR) per user, is illustrated in Table~\ref{SNR} for the different number of users. Specifically, the performance of Algorithm~\ref{alg_Find_B2} is compared with the two previous benchmarks:  $i)$ \textit{Beam Aperture} that was used in \cite[Algorithm 1]{pachler2021static} based  the angular separation between users and the beam aperture angle. Notice that this benchmark focuses on the beam placement only; and $ii)$ \textit{Homo. Balance} where nearest users are grouped together and served by the same beam. This idea was brought from the nearest user selection in terrestrial communications such as \cite{Chien2016b} and more focusing on the load balancing.  %In order to evaluate the advantages of our proposed algorithm, we consider two configurations for comparison as ; Beams serve an equal number of users and beam centers are obtained as in \eqref{eq_update_beam_center}. It is noted that to be able to execute this algorithm, the required number of active beams is found based on running stage 1 of Algorithm 1.
	Slight decreases of both the min CNR and average CNR per user are observed for all three algorithms as the number of users increases. However, Algorithm~\ref{alg_Find_B2} always achieves the best performance among the benchmarks. %The results demonstrates the benefits of jointly optimizing the beam placement and load balancing among the beams.  %For example, in the considered system including 25 users, Edge Based and Homo. Balance approaches results in 17.25 [dB] and 18.98 [dB] in terms of Average SNR, respectively, while the number for the Algorithm 1 is up to 1.8 [dB] higher.
	%\begin{table}[]
	%		\caption{SIMULATION PARAMETERS}
	%		\begin{tabular}{l|l}
	%			\hline
	%			Parameter                & Value \\
	%			\hline
	%			Bandwidth                &       \\
	%			The satellite's altitude [m]&   8063000   \\
	%			The total power at satellite [dB]& 23.5 \\
	%			     Downlink frequency [GHz]           &   18.050    \\
	%			     Antenna diameter [m]           &  0.6     \\
	%			 	Terminal antenna efficiancy		&  0.6    \\
	%			 	Aperture radius $(\beta)$  		&  $5 \lambda$    \\
	%			 	Noise power per user [dB]		&       \\
	%			 	Maximum gain at beam center		& 50 \\
	%			 	
	%			\hline
	%		\end{tabular}
	%	\centering
	%	\end{table}
	%	\begin{figure}[t]
	%		\centering
	%		\includegraphics[width= 0.5 \textwidth]{figures/user_locatons.eps}
	%		\vspace*{-0.3cm}
	%		\caption{The user locations.}
	%		\label{fig:user_locations}
	%		\vspace*{-0.45cm}
	%	\end{figure}
	\vspace*{-0.2cm}
	\section{Conclusion}
	\vspace*{-0.1cm}
	This paper has demonstrated the critical role of beam placement optimization in order to offer good channel gains to all the available users in the coverage area. We proposed an approximate solution to the  non-convex optimization problem and obtained a good feasible solution in polynomial time by a two-stage algorithm. {\hili The first stage defines the minimal number of active beams based on the graph theory. Meanwhile, the second stage balances the load among the active beams by inheriting the benefits of unsupervised learning from K-means clustering.}  Numerical results showed good channel gains by the proposed beam placement and load balancing solution. % and opened up new room for jointly placing the beams and other radio resources such as beamforming and precoding vectors.  
	\vspace*{-0.15cm}
	%	\vspace{-0.3cm}
%	\setstretch{0.95}
	\bibliographystyle{IEEEtran}
	\balance
	\bibliography{Journal}
%	\vspace{0.15cm}
\end{document}